\documentclass[twoside,12pt]{article}%
\usepackage{CJK}
\usepackage{indentfirst}
\usepackage{bm}
\usepackage{graphicx}
\usepackage{epsfig}
\usepackage{amsmath}
\usepackage{amsfonts}
\usepackage{amssymb}
\usepackage{amsthm}
\usepackage{latexsym}
\newcommand{\tabincell}[2]{\begin{tabular}{@{}#1@{}}#2\end{tabular}}
\newtheorem{thm}{Theorem}
\newtheorem{lemma}{Lemma}
\usepackage{epsfig}
\usepackage{amsmath}
\usepackage{epstopdf}
\usepackage{pgf,fancyhdr}
\usepackage{float}
\begin{document}

\begin{center}
\LARGE\bf Detection of Genuine Multipartite Entanglement in Multipartite Systems 
\end{center}


\begin{center}
\rm  Jing Yun Zhao,$^1$ \  Hui Zhao,$^1$ \ Naihuan Jing,$^{2,3}$ \ and Shao-Ming Fei$^{4,5}$ 
\end{center}

\begin{center}
\begin{footnotesize} \sl
$^1$ College of Applied Sciences, Beijing University of Technology, Beijing 100124, China

$^2$ Department of Mathematics, North Carolina State University, Raleigh, NC 27695, USA

$^3$ Department of Mathematics, Shanghai University, Shanghai 200444, China

$^4$ School of Mathematical Sciences, Capital Normal University, Beijing 100048, China

$^5$ Max-Planck-Institute for Mathematics in the Sciences, 04103 Leipzig, Germany

\end{footnotesize}
\end{center}

\vspace*{2mm}

\begin{center}
\begin{minipage}{15.5cm}
\parindent 20pt\footnotesize

We investigate genuine multipartite entanglement in general multipartite systems.
Based on the norms of the correlation tensors of a multipartite state under various partitions, we present an analytical sufficient criterion for detecting the
genuine four-partite entanglement. The results are generalized to arbitrary multipartite systems.
\end{minipage}
\end{center}

\begin{center}
\begin{minipage}{15.5cm}
\begin{minipage}[t]{2.3cm}{\bf Keywords:}\end{minipage}
\begin{minipage}[t]{13.1cm}
genuine multipartite entanglement, correlation tensor
\end{minipage}\par\vglue8pt
{\bf PACS: } 03.65.Ud, 03.67.Mn
\end{minipage}
\end{center}

\section{Introduction}
Quantum entanglement is one of the most fascinating features in quantum physics, with numerous applications in quantum information processing,
secure communication and channel protocols [1,2,3]. In particular, the genuine multipartite entanglement appears to have
more significant advantages than the bipartite ones in these quantum tasks [4].

The notion of genuine multipartite entanglement (GME) was introduced in [5].
Let $H_i^d$, $i=1,2,...,n$, denote $d$-dimensional Hilbert spaces. An $n$-partite state $\rho  \in H_1^d \otimes ... \otimes H_n^d$ can be expressed as
$
\rho = \sum {{p_\alpha }} \left| {{\psi _\alpha }} \right\rangle \left\langle {{\psi_\alpha }} \right|,
$
where $0<p_\alpha\leq 1$, $\sum {{p_\alpha }}  = 1$, $\left| {{\psi _\alpha }} \right\rangle \in H_1^d \otimes ...\otimes H_n^d$ are normalized pure states.
$\rho$ is said to be fully separable if it can be written as
$\rho=\sum_i q_i\ \rho^1_i\otimes\rho^2_i\otimes\cdots\otimes\rho^n_i$,
where $q_i$ is a probability distribution and
$\rho^j_i$ are density matrices with respect to the subsystem $H_j$.
On the other hand, $\rho$ is called genuine $n$-partite entangled if $\left| {{\psi _\alpha }} \right\rangle$ are not separable under any bipartite partitions.

The genuine multipartite entangled states exist in physical systems like the ground state of the XY model [6]. However, it is extremely difficult to identify the GME for general mixed multipartite states.
The GME concurrence and its lower bound were studied in [7-9]. Some sufficient or necessary conditions of GME were presented in [10-12].
As for detection of GME, the common criterion is the entanglement witnesses [13-16]. Using correlation tensors, the authors in [17]
have provided a general framework to detect different classes of GME for quantum systems of arbitrary dimensions.
In [18] the genuine multipartite entanglement has been investigated in terms of the norms of the correlation tensors and multipartite concurrence.
The relations between the norms of the correlation tensors and the detection of GME in tripartite quantum systems have been established in [19].

We need to use some simple mathematical concepts in this paper, let's briefly review them here. The elements of a vector space are called vectors. As we known, tensor product is a way of putting vector spaces together to form larger vector spaces. Suppose $W$ and $V$ are Hilbert spaces of dimension $m$ and $n$ respectively. Then $W\otimes V$ is an $mn$ dimensional vector space. The elements of $W\otimes V$ are liner combinations of `tensor products' $u\otimes v$ of elements $u$ of $W$ and $v$ of $V$. The outer product of $u$ and $v$ is equivalent to a matrix multiplication $uv^t$, provided that $u$ is represented as a $m\times 1$ column vector and $v$ as a $n\times 1$ column vector (which makes $v^t$ a row vector).

In this paper, we analyze the relationship between the norms of the correlation tensors and various bipartitions of multipartite quantum systems, and present sufficient conditions of GME for four partite and multipartite quantum systems.

We generalize some inequalities of the norms of the correlation tensors for four-partite states and give a criterion to detect GME of four-partite quantum systems in Section 2. In Section 3, we generalize these concepts and conclusions to multipartite quantum systems. Comments and conclusions are given in Section 4.

\section{Detection of GME for Four-partite Quantum States}
We first consider the GME for four-partite qudit states $\rho  \in H_1^d \otimes ... \otimes H_4^d$.
Let $\lambda_i$, $i=1,\cdots,d^2-1$, denote the mutually orthogonal generators of the special unitary Lie algebra $\mathfrak{su}(d)$ under
a fixed bilinear form [20], and $I$ the $d\times d$ identity matrix. Then $\rho$ can be expanded in terms of $\lambda_i$s,
\begin{eqnarray}
\rho&=&\frac{1}{d^4}I\otimes I\otimes I \otimes I+\frac{1}{2d^3}\sum^{4}_{f=1}\sum^{d^2-1}_{i_1=1} t_{i_1}^{(f)} \lambda_{i_1}^{(f)}\otimes I\otimes I \otimes I+\cdots\nonumber\\&&+\frac{1}{16}\sum^{d^2-1}_{i_1,i_2,i_3,i_4=1} t_{i_1,i_2,i_3,i_4}^{(1,2,3,4)} \lambda_{i_1}^{(1)}\otimes\lambda_{i_2}^{(2)} \otimes\lambda_{i_3}^{(3)} \otimes\lambda_{i_4}^{(4)},
\label{C}
\end{eqnarray}
where $\lambda_{i_1}^{(f)}$($(f)$ represents the position of $\lambda_{i_1}$ in the tensor product) stand for the operators with $\lambda_{i_1}$ on $H_{f}$ and $I$ on the rest spaces,
$t_{i_1}^{(f)}=tr(\rho\,\lambda_{i_1}^{(f)}\otimes I\otimes I \otimes I)$, $\cdots, t_{i_1,i_2,i_3,i_4}^{(1,2,3,4)}=tr(\rho\,\lambda_{i_1}^{(1)}\otimes \lambda_{i_2}^{(2)}\otimes \lambda_{i_3}^{(3)} \otimes \lambda_{i_4}^{(4)} )$.

Let $T^{(f)},\cdots,T^{(1,2,3,4)}$ denote  vectors with entries $t_{i_1}^{(f)},\cdots, t_{i_1,i_2,i_3,i_4}^{(1,2,3,4)}$$(i_1,i_2,i_3, i_4$ $=1,\cdots, d^2-1;f=1,2,3,4)$, respectively. From $T^{(f)},\cdots,T^{(1,2,3,4)}$ we further define the following matrices under different partitions.

We denote $T_{f|ghl}$ the ${(d^2-1)\times(d^2-1)^3}$ matrices with entries $t_{i_f,(d^2-1)^2(i_g-1)+(d^2-1)(i_h-1)+i_l}=t_{i_1,i_2,i_3,i_4}^{(1,2,3,4)}$, $T_{fg|hl}$ the ${(d^2-1)^2\times(d^2-1)^2}$ matrices with entries $t_{(d^2-1)(i_f-1)+i_g,(d^2-1)(i_h-1)+i_l}=t_{i_1,i_2,i_3,i_4}^{(1,2,3,4)}$, $T_{fgh|l}$ the ${(d^2-1)^3\times(d^2-1)}$ matrices with entries $t_{(d^2-1)^2(i_f-1)+(d^2-1)(i_g-1)+i_h,i_l}=t_{i_1,i_2,i_3,i_4}^{(1,2,3,4)}$, where $f\neq g\neq h\neq l=1,2,3,4; i_f,i_g,i_h,i_l=1,\cdots,d^2-1$. If the state is fully separable, we denote $T_{1|2|3|4}$ the ${(d^2-1)\times(d^2-1)^3}$ matrices with entries $t_{i_1,(d^2-1)^2(i_2-1)+(d^2-1)(i_3-1)+i_4}=t_{i_1,i_2,i_3,i_4}^{(1,2,3,4)}$.

Let $T^{(\underline{f},g)}$ and $T^{(f,\underline{g})}$ be ${(d^2-1)\times(d^2-1)}$ matrices with entries $t_{i_1,i_2}=t_{i_1,i_2}^{(f,g)}$ and $t_{i_2,i_1}=t_{i_1,i_2}^{(f,g)}$, respectively. We denote  $T^{(\underline{f},g,h)}$, $T^{(f,\underline{g},h)}$ and $T^{(f,g,\underline{h})}$  the ${(d^2-1)\times(d^2-1)^2}$ matrices with entries given by $t_{i_1,(d^2-1)(i_2-1)+i_3}={t_{i_1,i_2,i_3}^{(f,g,h)}}$, $t_{i_2,(d^2-1)(i_1-1)+i_3}={t_{i_1,i_2,i_3}^{(f,g,h)}}$ and $t_{i_3,(d^2-1)(i_1-1)+i_2}={t_{i_1,i_2,i_3}^{(f,g,h)}}$, respectively. We denote $T^{(\underline{f,g},h)}$, $T^{(\underline{f},g,\underline{h})}$ and $T^{(f,\underline{g,h})}$ the ${(d^2-1)^2\times(d^2-1)}$ matrices with entries given by $t_{(d^2-1)(i_1-1)+i_2,i_3}=t_{i_1,i_2,i_3}^{(f,g,h)}$, $t_{(d^2-1)(i_1-1)+i_3,i_2}=t_{i_1,i_2,i_3}^{(f,g,h)}$ and $t_{(d^2-1)(i_2-1)+i_3,i_1}=t_{i_1,i_2,i_3}^{(f,g,h)}$, respectively.

The Frobenius norm is matrix norm of an $m\times n$ matrix $M$ defined as the square root of the sum of the absolute squares of its elements, $\parallel M\parallel$=$\sqrt{\sum_{i,j}|M_{ij}|^2}$. It is also equal to the square root of the matrix trace of $MM^\dag$, where $M^\dag$ is the conjugate transpose, i.e., $\parallel M\parallel=\sqrt{tr(MM^\dag)}$. Since trace is invariant under unitary equivalence, this shows $\parallel M\parallel=\sqrt{\sum_i\sigma_i^2}$. The sum of the $k$ largest singular values of $M$ is a matrix norm, the Ky Fan $k$-norm of $M$, i.e., $\parallel M\parallel_k=\sum_i^k\sigma_i$, where $\sigma_i$, $i=1,\cdots,min(m,n)$, are the singular values of the matrix $M$ arranged in descending order.

For any pure state $\rho\in H^d_1\otimes H^d_2 \otimes H^d_3$,
$\rho=\frac{1}{d^3}I\otimes I\otimes I +\frac{1}{2d^2}(\sum^{d^2-1}_{i_1} t_{i_1}^{(1)} \lambda_{i_1}^{(1)}\otimes I\otimes I+\sum^{d^2-1}_{i_2} t_{i_2}^{(2)} I\otimes \lambda_{i_2}^{(2)}\otimes I+\sum^{d^2-1}_{i_3} t_{i_3}^{(3)} I \otimes I\otimes \lambda_{i_3}^{(3)})+\frac{1}{4d}(\sum^{d^2-1}_{i_1,i_2} t_{i_1,i_2}^{(1,2)} \lambda_{i_1}^{(1)}\otimes\lambda_{i_2}^{(2)} \otimes I+\sum^{d^2-1}_{i_2,i_3} t_{i_2,i_3}^{(2,3)} I\otimes\lambda_{i_2}^{(2)} \otimes \lambda_{i_3}^{(3)}+\sum^{d^2-1}_{i_1,i_3} t_{i_1,i_3}^{(1,3)} \lambda_{i_1}^{(1)}\otimes I  \otimes \lambda_{i_3}^{(3)})+\frac{1}{8}\sum^{d^2-1}_{i_1,i_2,i_3} t_{i_1,i_2,i_3}^{(1,2,3)} \lambda_{i_1}^{(1)}\otimes\lambda_{i_2}^{(2)} \otimes \lambda_{i_3}^{(3)}$, we have
$tr(\rho^2)=\frac{1}{d^3}+\frac{1}{2d^2}[\sum(t_{i_1}^{(1)})^2+\sum(t_{i_2}^{(2)})^2+
\sum(t_{i_3}^{(3)})^2]+\frac{1}{4d}[\sum(t_{i_1,i_2}^{(1,2)})^2+\sum(t_{i_1,i_3}^{(1,3)})^2+
\sum(t_{i_2,i_3}^{(2,3)})^2]+\frac{1}{8}\sum(t_{i_1,i_2,i_3}^{(1,2,3)})^2=1$. Therefore
\begin{eqnarray}\nonumber
\sum(t_{i_1,i_2,i_3}^{(1,2,3)})^2&&=\frac{8(d^3-1)}{d^3}-\Big\{\frac{4}{d^2}\left[\sum(t_{i_1}^{(1)})^2+\sum(t_{i_2}^{(2)})^2+
\sum(t_{i_3}^{(3)})^2\right]+\nonumber\\&&~~\frac{2}{d}\left[\sum(t_{i_1,i_2}^{(1,2)})^2+ \sum(t_{i_1,i_3}^{(1,3)})^2+
\sum(t_{i_2,i_3}^{(2,3)})^2\right]\Big\}\nonumber\\&&\leq\frac{8(d^3-1)}{d^3}.\nonumber
\end{eqnarray}
Thus, $\parallel T^{(1,2,3)}\parallel=\sqrt{\sum(t_{i_1,i_2,i_3}^{(1,2,3)})^2}\leq\frac{2}{d}\sqrt{\frac{2(d^3-1)}{d}}$.
Concerning the relations between the correlation tensors and the separability under various partitions, we have the following results:

\begin{lemma}
Let $\rho\in H^d_1\otimes H^d_2 \otimes H^d_3 \otimes H^d_4$ be a pure state. If $\rho$ is fully separable, then for any $k=1,\cdots,d^2-1$,
\begin{eqnarray}
 \parallel T_{1|2|3|4}\parallel_k= {\frac{4(d-1)^2}{d^2}}.
\end{eqnarray}
\end{lemma}
\begin{proof}
Since $\rho$ is fully separable, $\rho=\rho_{1}\otimes \rho_{2}\otimes \rho_{3}\otimes \rho_{4}$, where $\rho_{1} ,\rho_{2},\rho_{3},\rho_{4}$ are the reduced density matrices of $\rho$. By the calculation, we obtain  $t_{i_1,i_2,i_3,i_4}^{(1,2,3,4)}=t_{i_1}^{(1)}t_{i_2}^{(2)}t_{i_3}^{(3)}t_{i_4}^{(4)}$.
According to the inequality for 1-body correlation tensors, $\parallel T^{(f)}\parallel\leq\sqrt{\frac{2(d-1)}{d}}$ [17], $f=1, 2, 3, 4$, with the
equality holding iff the state is pure, we have
\begin{eqnarray}
\parallel T_{1|2|3|4}\parallel_k&&=\parallel T^{(1)}(T^{(2)}\otimes T^{(3)}\otimes T^{(4)})^t\parallel_k=\parallel T^{(1)}\parallel\cdot\parallel( T^{(2)}\otimes T^{(3)}\otimes T^{(4)})^t\parallel_k\nonumber\\&&=\parallel T^{(1)}\parallel\cdot\parallel ( T^{(2)}\otimes T^{(3)}\otimes T^{(4)})^t\parallel=\parallel T^{(1)}\parallel\cdot\parallel T^{(2)}\otimes T^{(3)}\otimes T^{(4)}\parallel\nonumber\\&&=\parallel T^{(1)}\parallel\cdot\parallel T^{(2)}\parallel\cdot\parallel T^{(3)}\parallel\cdot\parallel T^{(4)}\parallel=\frac{4(d-1)^2}{d^2},
\end{eqnarray}
which proves the Theorem.
\end{proof}
Let $f$, $g$, $h$ and $l$ be any subsystem in a four-partite quantum system. $f\neq g\neq h\neq l\in \{1,2,3,4\}$ means that any two subsystems are not repeatedly selected.
\begin{lemma}
Let $\rho\in H^d_1\otimes H^d_2 \otimes H^d_3 \otimes H^d_4$ be a pure state such that $\rho$ is separable
under at least one bipartition. Then for any $k=1,\cdots,d^2-1$, and
$f\neq g\neq h\neq l\in \{1,2,3,4\}$, we have\\
(i) if $\rho$ is separable under bipartition $f|ghl$, then
\begin{eqnarray}
 \parallel T_{f|ghl}\parallel_k\leq {\frac{4(d-1)\sqrt{d^2+d+1}}{d^2}};
\end{eqnarray}
(ii) if $\rho$ is entangled under bipartition $f|ghl$, then
\begin{eqnarray}
 \parallel T_{f|ghl}\parallel_k\leq {\frac{4\sqrt{k}(d^2-1)}{d^2}}.
\end{eqnarray}
\end{lemma}
\begin{proof}
(i) If $\rho$ is separable under bipartition $f|ghl$, $\rho=\rho_{f}\otimes\rho_{ghl}$, it
follows from $\parallel T^{(f,g,h)}\parallel\leq\frac{2}{d}\sqrt{\frac{2(d^3-1)}{d}}$ that
\begin{eqnarray}
\parallel T_{f|ghl}\parallel_k&&=\parallel T^{(f)}(T^{(g,h,l)})^t\parallel_k=\parallel T^{(f)}\parallel\cdot\parallel(T^{(g,h,l)})^t\parallel_k\nonumber\\&&=\parallel T^{(f)}\parallel\cdot\parallel (T ^{(g,h,l)})^t\parallel=\parallel T^{(f)}\parallel\cdot\parallel T^{(g,h,l)}\parallel\nonumber\\&&\leq\frac{4(d-1)\sqrt{d^2+d+1}}{d^2}.
\end{eqnarray}

(ii) $\rho$ is entangled under bipartition $f|ghl$, without loss of generality, say, under the bipartition $1|234$.
If $\rho$ is separable under some bipartition of one subsystem vs the rest three subsystems, we have
\begin{eqnarray}
\parallel T_{f|ghl}\parallel_k\leq\frac{4(d-1)\sqrt{d^2+d+1}}{d^2}.
\end{eqnarray}
If $\rho$ is separable under some bipartition of two subsystems vs the rest two subsystems,  from the inequality of 2-body correlation tensors $\parallel T^{(f,g)}\parallel\leq\sqrt{\frac{4(d^2-1)}{d^2}}$ [17],  we have
\begin{eqnarray}
\parallel T_{f|ghl}\parallel_k&&=\parallel T^{(\underline{f},{g})}\otimes(T^{({h},{l})})^t\parallel_k=\parallel T^{(\underline{f},{g})}\parallel_k\cdot\parallel (T^{({h},{l})})^t\parallel_k\nonumber\\&&\leq\sqrt{k}\parallel T^{(\underline{f},{g})}\parallel\cdot\parallel T^{({h},{l})}\parallel\leq{\frac{4\sqrt{k}(d^2-1)}{d^2}},
\end{eqnarray}
where we have used the inequality $\parallel M\parallel_k\leq k\parallel M\parallel$ for any matrix M.
If $\rho$ is separable under some bipartition of three  subsystems vs the rest one subsystem, we have
\begin{eqnarray}
\parallel T_{f|ghl}\parallel_k&&=\parallel T^{(\underline{f},g,h)}\otimes(T^{(l)})^t\parallel_k=\parallel T^{(\underline{f},g,h)}\parallel_k\cdot\parallel (T^{(l)})^t\parallel_k\nonumber\\&&\leq\sqrt{k}\parallel T^{(\underline{f},g,h)}\parallel\cdot\parallel T^{(l)}\parallel\leq\frac{4(d-1)\sqrt{k(d^2+d+1)}}{d^2}.
\end{eqnarray}
Hence, if $\rho$ is entangled under bipartition $1|234$, we have
$\parallel T_{f|ghl}\parallel_k\leq max\{\frac{4(d-1)\sqrt{d^2+d+1}}{d^2}, \frac{4\sqrt{k}(d^2-1)}{d^2},$\\
$\frac{4(d-1)\sqrt{k(d^2+d+1)}}{d^2}\}$ $={\frac{4\sqrt{k}(d^2-1)}{d^2}}$. Similar discussion applies to other bipartitions $2|134,$ $3|124$ and $4|123$. It indicates that these norms have the same upper bound. Hence, $\parallel T_{f|ghl}\parallel_k\leq {\frac{4\sqrt{k}(d^2-1)}{d^2}}$, if $\rho$ is entangled under bipartition $f|ghl$.
\end{proof}

We may analyze the bipartition $fgh|l$ by using similar methods above and obtain the following Lemma.

\begin{lemma}
Let $\rho\in H^d_1\otimes H^d_2 \otimes H^d_3 \otimes H^d_4$ be a pure state such that $\rho$ is separable
under at least one bipartition. Then for any $k=1,\cdots,d^2-1$, and
$f\neq g\neq h\neq l\in \{1,2,3,4\}$, we have\\
(i) if $\rho$ is separable under bipartition $fgh|l$, then
\begin{eqnarray}
 \parallel T_{fgh|l}\parallel_k\leq  {\frac{4(d-1)\sqrt{d^2+d+1}}{d^2}};
\end{eqnarray}
(ii) if $\rho$ is entangled under bipartition $fgh|l$, then
\begin{eqnarray}
\parallel T_{fgh|l}\parallel_k\leq{\frac{4\sqrt{k}(d^2-1)}{d^2}}.
\end{eqnarray}
\end{lemma}
Now we consider the relations between the correlation tensors and the separability under the bipartition $fg|hl$.
\begin{lemma}
Let $\rho\in H^d_1\otimes H^d_2 \otimes H^d_3 \otimes H^d_4$ be a pure state such that $\rho$ is separable
under at least one bipartition. Then for any $k=1,\cdots,d^2-1$, and
$f\neq g\neq h\neq l\in \{1,2,3,4\}$, we have\\
(i) if $\rho$ is separable under bipartition $fg|hl$, then
\begin{eqnarray}
 \parallel T_{fg|hl}\parallel_k\leq {\frac{4(d^2-1)}{d^2}};
\end{eqnarray}
(ii) if $\rho$ is entangled under bipartition $fg|hl$, then
\begin{eqnarray}
\parallel T_{fg|hl}\parallel_k\leq {\frac{4k(d^2-1)}{d^2}}.
\end{eqnarray}
\end{lemma}
\begin{proof}
(i) If $\rho$ is separable under bipartition $fg|hl$, $\rho=\rho_{fg}\otimes\rho_{hl}$, then
\begin{eqnarray}
\parallel T_{fg|hl}\parallel_k&&=\parallel T^{(f,g)}(T^{(h,l)})^t\parallel_k=\parallel T^{(f,g)}\parallel\cdot\parallel(T^{(h,l)})^t\parallel_k=\parallel T^{(f,g)}\parallel\cdot\parallel (T^{(h,l)})^t\parallel\nonumber\\&&=\parallel T^{(f,g)}\parallel\cdot\parallel T^{(h,l)}\parallel\leq\frac{4(d^2-1)}{d^2},
\end{eqnarray}
by using the inequality for 2-body correlation tensors.

(ii) $\rho$ is entangled under bipartition $fg|hl$, say, $12|34$. If $\rho$ is separable under some bipartition of one subsystem vs the rest three subsystems, we have
\begin{eqnarray}
\parallel T_{fg|hl}\parallel_k&&=\parallel T^{(f)}\otimes T^{(\underline{{g}},h,l)}\parallel_k=\parallel T^{(f)}\parallel\cdot\parallel T^{(\underline{g},h,l)}\parallel_k\nonumber\\&&\leq\sqrt{k}\parallel T^{(f)}\parallel\cdot\parallel T^{(\underline{g},h,l)}\parallel\leq\frac{4(d-1)\sqrt{k(d^2+d+1)}}{d^2}.
\end{eqnarray}
If $\rho$ is separable under some bipartition of two subsystems vs the rest two subsystems, we have
\begin{eqnarray}
\parallel T_{fg|hl}\parallel_k\leq{\frac{4(d^2-1)}{d^2}}.
\end{eqnarray}
If $\rho$ is separable under some bipartition of three  subsystems vs the rest one subsystem, we have
\begin{eqnarray}
\parallel T_{fg|hl}\parallel_k&&=\parallel T^{(\underline{f,g},h)}\otimes(T^{(l)})^t\parallel_k=\parallel T^{(\underline{f,g},h)}\parallel_k\cdot\parallel (T^{(l)})^t\parallel_k\nonumber\\&&\leq\sqrt{k}\parallel T^{(\underline{f,g},h)}\parallel\cdot\parallel T^{(l)}\parallel\leq\frac{4(d-1)\sqrt{k(d^2+d+1)}}{d^2}.
\end{eqnarray}
Hence, if $\rho$ is entangled under bipartition $12|34$, we have $\parallel T_{fg|hl}\parallel_k\leq$max$ \{\frac{4(d-1)\sqrt{k(d^2+d+1)}}{d^2},\\ \frac{4(d^2-1)}{d^2}\}=\frac{4(d-1)\sqrt{k(d^2+d+1)}}{d^2},$ $k\geq2$.
If $k=1$, $\parallel T_{fg|hl}\parallel_1\leq\frac{4(d^2-1)}{d^2}$.

Similarly, if $\rho$ is entangled under bipartition $13|24$, $14|23$  $23|14$, $24|13$ and $34|12$, we have the upper bound of the norm  as follows.
Let $i\ vs\ j$ denote that $\rho$ is separable under some bipartition of $i$ subsystem vs the rest $j$ subsystems.

\begin{tabular}{|c|c|c|c|c|c|}
\hline
  & $1\ vs\ 3$ & $2\ vs\ 2$  &  $3\ vs\ 1$   \\
\hline
$13|24$ & \tabincell{c}{$\parallel T_{13|24}\parallel_k~~~~~~~~~~~$\\$=\parallel T^{(f)}\otimes T^{(g,\underline{h},l)}\parallel_k$ \\$\leq\frac{4(d-1)\sqrt{k(d^2+d+1)}}{d^2}~~~~~~$ }
& \tabincell{c}{$\parallel T_{13|24}\parallel_k~~~~~~~~~~$\\$=\parallel T^{(\underline{{f}},g)}\otimes T^{(\underline{h},l)}\parallel_k$ \\$\leq\frac{4k(d^2-1)}{d^2}~~~~~~~~~~~~~~~$ } &
\tabincell{c}{$\parallel T_{13|24}\parallel_k~~~~~~~~~~~~~~$\\$=\parallel T^{(\underline{f},g,\underline{h})}\otimes (T^{(l)})^t\parallel_k$ \\$\leq\frac{4(d-1)\sqrt{k(d^2+d+1)}}{d^2}~~~~~~~~~$ }  \\
\hline
$14|23$ & \tabincell{c}{$\parallel T_{14|23}\parallel_k~~~~~~~~~$\\$=\parallel T^{(f)}\otimes T^{(g,h,\underline{l})}\parallel_k$ \\$\leq\frac{4(d-1)\sqrt{k(d^2+d+1)}}{d^2}~~~~~$ }
& \tabincell{c}{$\parallel T_{14|23}\parallel_k~~~~~~~~~$\\$=\parallel T^{(\underline{f},g)}\otimes T^{(h,\underline{l})}\parallel_k$ \\$\leq\frac{4k(d^2-1)}{d^2}~~~~~~~~~~~~~~~$ } &
\tabincell{c}{$\parallel T_{14|23}\parallel_k~~~~~~~~~~~~$\\$=\parallel T^{(\underline{f},g,h)}\otimes T^{(l)}\parallel_k$ \\$\leq\frac{4(d-1)\sqrt{k(d^2+d+1)}}{d^2}~~~~~$ }  \\
\hline
$23|14$ & \tabincell{c}{$\parallel T_{23|14}\parallel_k~~~~~~~~~~~$\\$=\parallel T^{(f)^t}\otimes T^{(\underline{g,h},l)}\parallel_k$ \\$\leq\frac{4(d-1)\sqrt{k(d^2+d+1)}}{d^2}~~~~~$ }
& \tabincell{c}{$\parallel T_{23|14}\parallel_k~~~~~~~~~~$\\$=\parallel T^{(f,\underline{g})}\otimes T^{(\underline{h},l)}\parallel_k$ \\$\leq\frac{4k(d^2-1)}{d^2}~~~~~~~~~~~~~~~$ } &
\tabincell{c}{$\parallel T_{23|14}\parallel_k~~~~~~~~~~~~~~$\\$=\parallel T^{(f,\underline{g,h})}\otimes (T^{(l)})^t\parallel_k$ \\$\leq\frac{4(d-1)\sqrt{k(d^2+d+1)}}{d^2}~~~~~~~~$ }  \\
\hline
$24|13$ & \tabincell{c}{$\parallel T_{24|13}\parallel_k~~~~~~~~~~$\\$=\parallel T^{(f)^t}\otimes T^{(\underline{g},h,\underline{l})}\parallel_k$ \\$\leq\frac{4(d-1)\sqrt{k(d^2+d+1)}}{d^2}~~~~~$ }
& \tabincell{c}{$\parallel T_{24|13}\parallel_k~~~~~~~~~~~$\\$=\parallel T^{(f,\underline{g})}\otimes T^{(h,\underline{l})}\parallel_k$ \\$\leq\frac{4k(d^2-1)}{d^2}~~~~~~~~~~~~~~~$ } &
\tabincell{c}{$\parallel T_{24|13}\parallel_k~~~~~~~~~~~$\\$=\parallel T^{(f,\underline{g},h)}\otimes T^{(l)}\parallel_k$ \\$\leq\frac{4(d-1)\sqrt{k(d^2+d+1)}}{d^2}~~~~~$ }  \\
\hline
$34|12$   &\tabincell{c}{$\parallel T_{34|12}\parallel_k~~~~~~~~~~~~~~$ \\$=\parallel (T^{(f)})^t\otimes T^{(g,\underline{h,l})}\parallel_k$\\ $\leq\frac{4(d-1)\sqrt{k(d^2+d+1)}}{d^2}~~~~~~~~$ } & \tabincell{c}{$\parallel T_{34|12}\parallel_k~~~~~~~~~~~$\\$=\parallel (T^{f,g})^t\otimes T^{(h,l)}\parallel_k$ \\$\leq\frac{4(d^2-1)}{d^2}~~~~~~~~~~~~~~~~~~$ }   &\tabincell{c}{$\parallel T_{34|12}\parallel_k~~~~~~~~~~~$\\$=\parallel T^{(f,g,\underline{h})}\otimes T^{(l)}\parallel_k$ \\$\leq\frac{4(d-1)\sqrt{k(d^2+d+1)}}{d^2}~~~~~$ } \\
\hline
\end{tabular}\\
Altogether we have $\parallel T_{fg|hl}\parallel_k\leq {\frac{4k(d^2-1)}{d^2}}$ if $\rho$ is entangled under bipartition $fg|hl$.
\end{proof}
Next we  present a sufficient condition to detect GME for four-partite systems. By the Lemma 2 we have that $\parallel T_{f|ghl}\parallel_k\leq {\frac{4(d-1)\sqrt{d^2+d+1}}{d^2}}$ if $\rho$ is separable, and $\parallel T_{f|ghl}\parallel_k\leq{\frac{4\sqrt{k}(d^2-1)}{d^2}}$  if $\rho$ is entangled. However, $\parallel T_{fg|hl}\parallel_k\leq{\frac{4k(d^2-1)}{d^2}}$  is a rather weak condition.
We define the average matricization norm, $M_k=\frac{1}{4}(\parallel T_{1|234}\parallel_k+\parallel T_{2|134}\parallel_k+\parallel  T_{3|124}\parallel_k+\parallel T_{4|123}\parallel_k)$.

\begin{thm}
If $\rho$ is a four-qudit state, and
\begin{eqnarray}\label{thm1}
M_k(\rho)>\frac{(d-1)[\sqrt{d^2+d+1}+3(d+1)\sqrt{k}]}{d^2}
\end{eqnarray}
for any $k\in\{1,2,3,\cdots,d^2-1\}$, then $\rho$ is genuine multipartite entangled.
\label{1}
\end{thm}

{\sf Remark 1:} Compared with the Theorem 3 in [17] for four-qubit states, our result detects
GME for any general four-qudit states.

\section{Detection of GME for Multipartite Quantum States }
In this section, we study the GME for multipartite qudit states. Any n-partite density matrix $\rho\in H_1^d\otimes H_2^d\otimes \cdots\otimes H_n^d$ can be expressed as
\begin{eqnarray}
\rho&&=\frac{1}{d^n}I\otimes \cdots\otimes I+\frac{1}{2d^{n-1}}\sum_{j_1=1}^n\sum^{d^2-1}_{i_1=1} t_{i_1}^{(j_1)} \lambda_{i_1}^{(j_1)}\otimes I\otimes\cdots\otimes I+\cdots\nonumber\\&&~~
+\frac{1}{2^n}\sum_{i_1,\cdots,i_n=1}^{d^2-1} t_{i_1,\cdots,i_n}^{(1,\cdots,n)}\lambda_{i_1}^{(1)}\otimes \lambda_{i_2}^{(2)}\otimes \cdots\otimes \lambda_{i_n}^{(n)},
\end{eqnarray}
where $(j_1)$ represents the position of $\lambda_{i_1}$ in the tensor product, $t_{i_1}^{(j_1)}=tr(\rho\lambda_{i_1}^{(j_1)}\otimes I\otimes\cdots\otimes I), \cdots,t_{i_1,\cdots,i_n}^{(1,\cdots,n)}=tr(\rho\lambda_{i_1}^{(1)}\otimes \lambda_{i_2}^{(2)}\otimes \cdots\otimes \lambda_{i_n}^{(n)})$, and $T^{(j_1)},$ $\cdots,T^{(1,\cdots, n)}$ are the vectors (tensors) with elements $t_{i_1}^{(j_1)},\cdots,t_{i_1,\cdots,i_n}^{(1,\cdots,n)}$, respectively.

For a pure state $\rho$, one has
\begin{eqnarray}
tr(\rho^2)&&=\frac{1}{d^n}+\frac{1}{2d^{n-1}}\sum_{j_1}^n\sum_{i_1}^{d^2-1}(t_{i_1}^{(j_1)})^2+
\cdots+{\frac{1}{2^n}\sum_{i_1,\cdots,i_n}^{d^2-1}(t_{i_1,\cdots,i_n}^{(1,\cdots,n)})^2=1.}
\end{eqnarray}
Hence
\begin{eqnarray}
\sum_{i_1,\cdots,i_n}^{d^2-1}(t_{i_1,\cdots,i_n}^{(1,\cdots,n)})^2&&={2^n-\frac{2^n}{d^n}-\cdots-
\frac{2^n}{2d^{n-1}}\sum_{j_1}^n\sum_{i_1}^{d^2-1}(t_{i_1}^{(j_1)})^2\leq \frac{2^n(d^n-1)}{d^n},}
\end{eqnarray}
which implies that
\begin{equation}\label{l3}
\parallel T^{(1,2,\cdots,n)}\parallel=\sqrt{\sum_{i_1,\cdots,i_n}^{d^2-1}(t_{i_1,\cdots,i_n}^{(1,\cdots,n)})^2}
\leq\sqrt{\frac{2^n(d^n-1)}{d^n}}.
\end{equation}
We now consider multipartite systems and their $T$ matrices.

\begin{thm}
Let $\rho\in H^d_1\otimes\cdots \otimes H^d_n $ be a pure state. If $\rho$ is fully separable, then for any $k=1,\cdots,d^2-1$,
\begin{eqnarray}
 \parallel T_{1|\cdots|n}\parallel_k=\sqrt{\frac{2^n(d-1)^n}{d^n}}.
\end{eqnarray}
\end{thm}
\begin{proof}
According to the Proposition 1 of Ref. [21], i.e., if $\rho$ is fully separable then $t_{i_1,\cdots,i_n}^{(1,\cdots,n)}=t_{i_1}^{(1)}\cdots  t_{i_n}^{(n)}$, using the bound  $\parallel T^{(j_1)}\parallel\leq\sqrt{\frac{2(d-1)}{d}}$, $j_1=1,\cdots,n$, we have
\begin{eqnarray}
\parallel T_{1|\cdots|n}\parallel_k&&=\parallel T^{(1)}(T^{(2)}\otimes\cdots\otimes T^{(n)})^t \parallel_k=\parallel T^{(1)}\parallel\cdot\parallel(T^{(2)}\otimes\cdots\otimes T^{(n)})^t\parallel_k\nonumber\\&&=\parallel T^{(1)}\parallel\cdot\parallel T^{(2)}\otimes\cdots\otimes T^{(n)}\parallel_k=\parallel T^{(1)}\parallel\cdot\parallel T^{(2)}\otimes\cdots\otimes T^{(n)}\parallel\nonumber\\&&=\parallel T^{(1)}\parallel\cdot\parallel T^{(2)}\parallel\cdots \parallel T^{(n)}\parallel=\sqrt{\frac{2^n(d-1)^n}{d^n}}.
\end{eqnarray}
Hence, if $\rho$ is fully separable, then
$\parallel T_{1|\cdots|n}\parallel_k=\sqrt{\frac{2^n(d-1)^n}{d^n}}.$
\end{proof}

Let $A_1$ be subsets of the set $\{H_1, H_2,\cdots, H_n\}$ and $A_2$ the complement of $A_1$, $n_{A_1}$ and $n_{A_2}$
be the number of spaces contained in $A_1$ and $A_2$, respectively. For the bipartition $A_1|A_2=j_1\cdots j_{n_{A_1}}|j_{n_{A_1+1}}\cdots j_n$, $j_1\neq j_2\neq \cdots\neq j_n\in \{1,2,\cdots,n\}$ $($this means that any two subsystems are not repeatedly selected$)$,
let $T_{A_1|A_2}$ be a matrix with entries $t_{a,b}=t_{i_1,\cdots,i_n}^{(1,\cdots,n)}$, where $a=(d^2-1)^{n_{A_1}-1}(i_{j_1}-1)+\cdots+i_{j_{n_{A_1}}}$,
$b=(d^2-1)^{n_{A_2}-1}(i_{j_{n_{A_1}}+1}-1)+\cdots+i_{j_n}$, $i_{j_1}, i_{j_2}, \ldots, i_{j_n}=1, 2,\ldots, d^2-1$.

\begin{thm}
Let $\rho\in H^d_1\otimes\cdots \otimes H^d_n $ be a pure state. If $\rho$ is separable under bipartition $A_1|A_2$, then for any $k=1,\cdots,d^2-1$,
\begin{eqnarray}
\parallel T_{A_1|A_2}\parallel_k\leq \sqrt{\frac{2^n(d^{n_{A_1}}-1)(d^{n_{A_2}}-1)}{d^n}}.
\end{eqnarray}
\end{thm}
\begin{proof}
If $\rho$ is separable under bipartition $A_1|A_2$, then $\rho_{A_1}\otimes \rho_{A_2}$. Using the
inequality (22), we get
\begin{eqnarray}
\parallel T_{A_1|A_2}\parallel_k&&=\parallel T^{(A_1)}(T^{(A_2)})^t \parallel_k=\parallel T^{(A_1)}\parallel\cdot\parallel (T^{(A_2)})^t\parallel_k\nonumber\\&&=\parallel T^{(A_1)}\parallel\cdot\parallel (T^{(A_2)})^t\parallel=\parallel T^{(A_1)}\parallel\cdot\parallel (T^{(A_2)})\parallel\nonumber\\&&\leq\sqrt{\frac{2^n(d^{n_{A_1}}-1)(d^{n_{A_2}}-1)}{d^n}}.
\end{eqnarray}
\end{proof}

\begin{thm}
Let $\rho\in H^d_1\otimes\cdots \otimes H^d_n $ be a pure state such that $\rho$ is separable
under at least one bipartition. For any $k=1,\cdots,d^2-1$ and $j_1\neq j_2\neq \cdots\neq j_n\in \{1,2,\cdots,n\}$,  we have

$(i)$  if $\rho$ is entangled under a certain bipartition $j_1|j_2\cdots j_n$, then

$\parallel T_{j_1|j_2\cdots j_n}\parallel_k\leq \sqrt{\frac{2^nk(d^{[\frac{n}{2}]}-1)(d^{n-[\frac{n}{2}]}-1)}{d^n}}$ $([]$ denotes integer function$)$, when $n$ is odd;

$\parallel T_{j_1|j_2\cdots j_n}\parallel_k\leq\sqrt{\frac{2^nk(d^{\frac{n}{2}}-1)^2}{d^n}}$, when $n$ is even;

$(ii)$ if $\rho$ is entangled under a certain bipartition  $j_1\cdots j_{n-1} |j_n$, then

$\parallel T_{j_1\cdots j_{n-1} |j_n}\parallel_k\leq \sqrt{\frac{2^nk(d^{[\frac{n}{2}]}-1)(d^{n-[\frac{n}{2}]}-1)}{d^n}}$, when $n$ is odd;

$\parallel T_{j_1\cdots j_{n-1} |j_n}\parallel_k\leq\sqrt{\frac{2^nk(d^{\frac{n}{2}}-1)^2}{d^n}}$, when $n$ is even.
\end{thm}
\begin{proof}
$(i)$ If $\rho$ is entangled under bipartition $j_1|j_2\cdots j_n$, then there is at least one bipartition ${j'_1}\cdots {j'_p}|{j'_{p+1}}\cdots {j'_n}$ $(p=1,2\cdots,n-1)$ such that $\rho$ is separable. Let ${j'_1}\cdots {j'_p}|{j'_{p+1}}\cdots {j'_n}=A_1|A_2$, then $n_{A_1}=p$.

~$\textcircled{1}$ $j_1=1$.  If $p=1$ and $j'_1\neq1$, we have
\begin{eqnarray}\label{u1}
\parallel T_{j_1|j_2\cdots j_n}\parallel_k&=&\parallel T^{({j'_1})}(T^{({j'_2},\cdots,{j'_n})})^t\parallel_k=\parallel T^{({j'_1})}\parallel\cdot\parallel(T^{({j'_2},\cdots,{j'_n})})^t\parallel_k\nonumber \\ &=&\parallel T^{({j'_1})}\parallel\cdot\parallel T^{({j'_2},\cdots,{j'_n})}\parallel\leq\sqrt{\frac{2^n(d-1)(d^{n-1}-1)}{d^n}}.
\end{eqnarray}

If $p=2,3,\cdots,n-1$, we get
\begin{eqnarray}\label{u2}
\parallel T_{j_1|j_2\cdots j_n}\parallel_k&=&\parallel T^{(\underline{j'_1},\cdots,j'_p)}\otimes(T^{(j'_{p+1},\cdots,j'_n)})^t\parallel_k=\parallel T^{(\underline{j'_1},\cdots,j'_p)}\parallel_k\cdot\parallel(T^{(j'_{p+1},\cdots,j'_n)})^t\parallel_k\nonumber \\ &\leq&\sqrt{k}\parallel T^{(\underline{j'_1},\cdots,j'_p)}\parallel\cdot\parallel T^{(j'_{p+1},\cdots,j'_n)}\parallel\leq\sqrt{\frac{2^nk(d^{p}-1)(d^{n-p}-1)}{d^n}}.
\end{eqnarray}

~$\textcircled{2}$ $j_1=2,\cdots,n-1$. For any $p$ we have
\begin{eqnarray}\label{u3}
\parallel T_{j_1|j_2\cdots j_n}\parallel_k&=&\parallel T^{(j'_1,\cdots,\underline{j'_{j_1}},\cdots,j'_p)}\otimes(T^{(j'_{p+1},\cdots,j'_n)})^t\parallel_k=\parallel T^{(j'_1,\cdots,\underline{j'_{j_1}},\cdots,j'_p)}\parallel_k\cdot\parallel(T^{(j'_{p+1},\cdots,j'_n)})^t\parallel_k\nonumber \\ &\leq&\sqrt{k}\parallel T^{(j'_1,\cdots,\underline{j'_{j_1}},\cdots,j'_p)}\parallel\cdot\parallel T^{(j'_{p+1},\cdots,j'_n)}\parallel\leq\sqrt{\frac{2^nk(d^{p}-1)(d^{n-p}-1)}{d^n}}.
\end{eqnarray}

~$\textcircled{3}$  $j_1=n$. If $p=1,\cdots,n-2$, we have
\begin{eqnarray}\label{u5}
\parallel T_{j_1|j_2\cdots j_n}\parallel_k&=&\parallel (T^{(j'_1,\cdots,j'_p)})^t\otimes T^{(j'_{p+1},\cdots,\underline{j'_n})}\parallel_k=\parallel (T^{(j'_1,\cdots,j'_p)})^t\parallel_k\cdot\parallel T^{(j'_{p+1},\cdots,\underline{j'_n})}\parallel_k\nonumber \\ &\leq&\sqrt{k}\parallel T^{(j'_1,\cdots,j'_p)}\parallel\cdot\parallel T^{(j'_{p+1},\cdots,\underline{j'_n})}\parallel\leq\sqrt{\frac{2^nk(d^{p}-1)(d^{n-p}-1)}{d^n}}.
\end{eqnarray}

If $p=n-1$, we get
\begin{eqnarray}\label{u6}
\parallel T_{j_1|j_2\cdots j_n}\parallel_k&=&\parallel (T^{(j'_1,\cdots,j'_{n-1})})^t\otimes T^{(j'_n)}\parallel_k=\parallel (T^{(j'_1,\cdots,j'_{n-1})})^t\parallel\cdot\parallel T^{(j'_n)}\parallel_k\nonumber \\
&=&\parallel T^{(j'_1,\cdots,j'_{n-1})}\parallel\cdot\parallel T^{(j'_n)}\parallel\leq\sqrt{\frac{2^n(d-1)(d^{n-1}-1)}{d^n}}.
\end{eqnarray}

Now consider max$\{\sqrt{\frac{2^nk(d^{p}-1)(d^{n-p}-1)}{d^n}},\sqrt{\frac{2^n(d-1)(d^{n-1}-1)}{d^n}}\}$ $p=1,\cdots,n-1$.
Let  $y=(d^h-1)(d^{n-h}-1)$ $(h>0)$ be a continuous function. Then the maximal value is $y_{max}=(d^{\frac{n}{2}}-1)^2$.
If $n$ is odd, $\parallel T_{A_1|A_2}\parallel_k\leq
\sqrt{\frac{2^nk(d^{[\frac{n}{2}]}-1)(d^{n-[\frac{n}{2}]}-1)}{d^n}}$.
If $n$ is even, $\parallel T_{A_1|A_2}\parallel_k\leq\sqrt{\frac{2^nk(d^{\frac{n}{2}}-1)^2}{d^n}}$.

$(ii)$ If $\rho$ is entangled under bipartition $j_1\cdots j_{n-1}|j_n$, then there is at least one bipartition ${j'_1}\cdots {j'_p}|{j'_{p+1}}\cdots {j'_n}$ $p=1,2\cdots,n-1$, such that $\rho$ is separable. Similarly, let ${j'_1}\cdots {j'_p}|{j'_{p+1}}\cdots {j'_n}=A_1|A_2$, then $n_{A_1}=p$. The proof can be done in three cases.

~$\textcircled{1}$ $j_n=1$. If $p=1$, we have
\begin{eqnarray}
\parallel T_{j_1\cdots j_{n-1}|j_n}\parallel_k&=&\parallel (T^{(j'_1)})^t\otimes T^{(j'_2,\cdots,j'_n)}\parallel_k=\parallel (T^{(j'_1)})^t\parallel_k\cdot\parallel T^{(j'_2,\cdots,j'_n)}\parallel_k\nonumber \\ &=&\parallel T^{(j'_1)}\parallel\cdot\parallel T^{(j'_2,\cdots,j'_n)}\parallel\leq\sqrt{\frac{2^n(d-1)(d^{n-1}-1)}{d^n}}.
\end{eqnarray}

If $p=2,\cdots,n-1$, we get
\begin{eqnarray}
\parallel T_{j_1\cdots j_{n-1}|j_n}\parallel_k&=&\parallel T^{(j'_1,\underline{j'_2,\cdots,j'_p})}\otimes T^{(j'_{p+1},\cdots,j'_n)}\parallel_k=\parallel T^{(j'_1,\underline{j'_2,\cdots,j'_p})}\parallel_k\cdot\parallel T^{(j'_{p+1},\cdots,j'_n)}\parallel_k\nonumber \\ &\leq&\sqrt{k}\parallel T^{(j'_1,\underline{j'_2,\cdots,j'_p})}\parallel\cdot\parallel T^{(j'_{p+1},\cdots,j'_n)}\parallel\leq\sqrt{\frac{2^nk(d^{p}-1)(d^{n-p}-1)}{d^n}}.
\end{eqnarray}

~$\textcircled{2}$ $j_n=2,\cdots,n-1$. For any $p$ we have
\begin{eqnarray}
\parallel T_{j_1\cdots j_{n-1}|j_n}\parallel_k&=&\parallel T^{(\underline{{j'_1},\cdots},{j'_{j_n}},\underline{\cdots,{j'_p}})}\otimes T^{(j'_{p+1},\cdots,{j'_n})}\parallel_k=\parallel T^{(\underline{{j'_1},\cdots},{j'_{j_n}},\underline{\cdots,{j'_p}})}\parallel_k\cdot\parallel T^{(j'_{p+1},\cdots,{j'_n})}\parallel_k\nonumber \\ &\leq&\sqrt{k}\parallel T^{(\underline{{j'_1},\cdots},{j'_{j_n}},\underline{\cdots,{j'_p}})}\parallel\cdot\parallel T^{(j'_{p+1},\cdots,{j'_n})}\parallel\leq\sqrt{\frac{2^nk(d^{p}-1)(d^{n-p}-1)}{d^n}}.
\end{eqnarray}

~$\textcircled{3}$ $j_n=n$. If $p=1,\cdots,n-2$, we get
\begin{eqnarray}
\parallel T_{j_1\cdots j_{n-1}|j_n}\parallel_k&=&\parallel T^{({j'_1},\cdots,{j'_p})}\otimes T^{(\underline{{j'_{p+1}},\cdots,{j'_{n-1}}},{j'_n})}\parallel_k=\parallel T^{({j'_1},\cdots,{j'_p})}\parallel_k\cdot\parallel T^{(\underline{{j'_{p+1}},\cdots,j'_{n-1}},{j'_n})}\parallel_k\nonumber \\ &\leq&\sqrt{k}\parallel T^{({j'_1},\cdots,{j'_p})}\parallel\cdot\parallel T^{(\underline{{j'_{p+1}},\cdots,j'_{n-1}},{j'_n})}\parallel\leq\sqrt{\frac{2^nk(d^{p}-1)(d^{n-p}-1)}{d^n}}.
\end{eqnarray}

If $p=n-1$ and $j'_n\neq n$, we have
\begin{eqnarray}
\parallel T_{j_1\cdots j_{n-1}|j_n}\parallel_k&=&\parallel T^{(j'_1,\cdots,j'_{n-1})} (T^{(j'_n)})^t\parallel_k=\parallel T^{(j'_1,\cdots,j'_{n-1})}\parallel\cdot\parallel(T^{(j'_n)})^t\parallel_k\nonumber \\
&=&\parallel T^{(j'_1,\cdots,j'_{n-1})}\parallel\cdot\parallel T^{(j'_n)}\parallel\leq\sqrt{\frac{2^n(d-1)(d^{n-1}-1)}{d^n}}.
\end{eqnarray}
If $n$ is odd, $\parallel T_{j_1\cdots j_{n-1}|j_n}\parallel_k\leq
\sqrt{\frac{2^nk(d^{[\frac{n}{2}]}-1)(d^{n-[\frac{n}{2}]}-1)}{d^n}}$. If $n$ is even, $\parallel T_{j_1\cdots j_{n-1}|j_n}\parallel_k\leq\sqrt{\frac{2^nk(d^{\frac{n}{2}}-1)^2}{d^n}}$.
\end{proof}

\section{Conclusion}
We have studied genuine multipartite entanglement in four-partite and multipartite qudit quantum systems, and derived the relationship between the norms of the correlation tensors and the specific matrix $T$.
Based on these relations we have presented a criterion to detect GME in four-partite quantum systems. These results  are generalized to multipartite systems.
Our main results concern with special inequalities that bound the various norms of the correlation tensors, upon which our criterion is presented to detect GME in multipartite systems. These results can
help distinguishing genuine multipartite entangled states. Genuine multipartite entanglement plays significant roles in many quantum information processing.
Our approach and results may highlight further researches on the theory of genuine multipartite entanglement.

\vspace*{2mm}

$\textbf{Acknowledgments}$ This work is supported by the National Natural Science Foundation of China under grant Nos. 11101017, 11531004, 11726016 and 11675113,
and Simons Foundation under grant No. 523868, the NSF of Beijing under Grant No. KZ201810028042.

\end{document}